\documentclass{amsart}

\usepackage[left=3cm,top=3cm,right=3cm,bottom=3cm]{geometry}
\usepackage{amssymb, amsmath, amsthm, units}
\usepackage{graphicx}
\usepackage{enumerate}
\usepackage{color}

\DeclareGraphicsExtensions{.pdf,.jpg,.png}

\theoremstyle{plain}
\newtheorem{theorem}{Theorem}
\newtheorem{lemma}[theorem]{Lemma}
\newtheorem{proposition}[theorem]{Proposition}

\theoremstyle{definition}

\newcommand{\legendre}[2]{\genfrac{(}{)}{}{}{#1}{#2}}
\newcommand{\eps}{\varepsilon}

\newcommand{\EE}{\mathbb{E}}

\newcommand{\PP}{\mathbb{P}}

\newcommand{\Fp}{\mathbb{F}_p}
\newcommand{\Fq}{\mathbb{F}_q}

\newcommand{\abs}[1]{\left|#1\right|}

\begin{document}

\title{The number of rational points of hyperelliptic curves over subsets of finite fields}

\author{Kristina Nelson}
\address{\noindent Department of Mathematics, University of California, Berkeley, CA 94720-3840}
\email{krisn@math.berkeley.edu}

\author{J\'{o}zsef Solymosi}
\address{\noindent Department of Mathematics, University of British Columbia, Vancouver, BC, Canada V6T 1Z2}
\thanks{The second author was supported by NSERC and Hungarian National Research Development and Innovation Fund K 119528}
\email{solymosi@math.ubc.ca}

\author{Foster Tom}
\address{\noindent Department of Mathematics, University of British Columbia, Vancouver, BC, Canada V6T 1Z2}
\email{foster@math.ubc.ca}

\author{Ching Wong}
\address{\noindent Department of Mathematics, University of British Columbia, Vancouver, BC, Canada V6T 1Z2}
\email{ching@math.ubc.ca}

\date{}

\begin{abstract}
We prove two related concentration inequalities concerning the number of rational points of hyperelliptic curves over subsets of a finite field. In particular, we investigate the probability of a large discrepancy between the numbers of quadratic residues and non-residues in the image of such subsets over uniformly random hyperelliptic curves of given degrees. We find a constant probability of such a high difference and show the existence of sets with an exceptionally large discrepancy.
\end{abstract}

\maketitle

\section{Introduction}
Let $q$ be a prime power and let $\Fq$ be the finite field with $q$ elements. A curve $E:y^2=f(x)$ (together with a point of infinity $\mathcal{O}$) is called an \emph{elliptic curve} over $\Fq$ if $f(x)\in\Fq[x]$ is a cubic polynomial having distinct roots in the algebraic closure $\overline{\Fq}$ of $\Fq$. The set of \emph{rational points} of $E$ in $\Fq$ is
\[
E(\Fq)=\{(x,y)\in\Fq\times\Fq:y^2=f(x)\}\cup\{\mathcal{O}\}.
\]

Suppose that $q$ is odd. Using the fact that there are $(q-1)/2$ invertible quadratic residues and $(q-1)/2$ non-residues in $\Fq$, one can approximate the size of $E(\Fq)$ as follows. For each $x\in\Fq$, the probability of $f(x)$ being a non-zero square in $\Fq$, and hence contributing 2 points to $E(\Fq)$, is about $1/2$. With probability about $1/2$ there is no point in $E(\Fq)$ having the first coordinate $x\in\Fq$. Therefore, $\#E(\Fq)$ is expected to be close to $q+1$. Indeed, Hasse \cite{Has36} proved, in 1936, that the error in this estimate is at most $2\sqrt{q}$:
\[
|\#E(\Fq)-(q+1)|\leq2\sqrt{q}.
\]

Knowledge of $\#E(\Fq)$ is crucial in elliptic curve cryptography (ECC), which is considered to be more efficient than the classical cryptosystems, like RSA \cite{RSA78}. The security of ECC depends on the difficulty of solving the Elliptic Curve Discrete Logarithm Problem (ECDLP). The best known algorithm to solve ECDLP in finite fields is Pollard's Rho Algorithm \cite{Pol75}, which requires $O(\sqrt{p})$ time complexity, where $p$ is the prime factor of $q$. However, some well studied elliptic curves or elliptic curves of certain forms are not good candidates for ECC. For instance, if the number of rational points of an elliptic curve $E$ in $\Fp$ is exactly $p$, where $p$ is a prime, then the running time of solving the ECDLP is $O(\log{p})$, see \cite{Sem98}. Using verifiably random elliptic curves in ECC can ensure higher security. Hyperelliptic curves can also be used in cryptography, see \cite{handbook06} for more details; however, the verifiability of random hyperelliptic curves is much harder, see \cite{HSS01,Sat09}.

In this paper, we investigate the behaviour of random hyperelliptic curves over subsets $S$ of $\Fq$. We are interested in the hyperelliptic curves $E:y^2=f(x)$ where $f(x)$ is a polynomial in $\Fq[x]$ of degree $4k-1$ ($k\geq1$) having distinct roots in $\overline{\Fq}$. Denote by $E(\Fq,S)$ the rational points of $E$ in $\Fq$ where the $x$-coordinate is in $S$, i.e.
\[
E(\Fq,S)=\{(x,y)\in{S}\times\Fq:y^2=f(x)\}.
\]
We remark that the point of infinity $\mathcal{O}$ is not included in $E(\Fq,S)$. The approximation we have described for $\#E(\Fq)$ suggests that the expected value of $\#E(\Fq,S)$ is about $\#S$. For random hyperelliptic curves $E$ over $\Fq$, the probability that the error $|\#E(\Fq,S)-\#S|$ is small has been extensively studied, see \cite{PelRam17,SSS95} for example. 

On the other hand, it is easy to see that there exist many hyperelliptic curves of any (positive) even degree so that the error $|\#E(\Fp,S)-\#S|$ is very large. Indeed, the error is about $\#S$ when $f(x)$ is the square of any non-constant polynomial in $\Fq[x]$, for any $S\subset{\Fp}$.

However, an error bound is not obvious in the case of hyperelliptic curves of odd degree, which we study in the probabilistic setting. Equivalently, we examine the difference between the numbers of quadratic residues and non-residues in the image multiset $f(S)$. Using $4k$-wise independence, we show that all subsets $S$ of $\Fq$ behave similarly, in the sense that the interested discrepancy is proportional to $\sqrt{\#S}$ has a positive probability which depends only on the degree of the curve.

\begin{theorem}
\label{thm:largeProb}
Given a positive integer $k$ and $\eps>0$, there exist $\delta>0$ and a threshold $N$ such that the following holds: for every odd prime power $q>N$, if a curve $E:y^2=f(x)$ is chosen uniformly at random among all degree $4k-1$ hyperelliptic curves over $\Fq$, then with a probability at least $(4\pi^{3/2}/e^3)2^{-2k}-\eps$, we have
\[
|\#E(\Fq,S)-\#S|>\delta\sqrt{\#S},
\]
for any set $S\subset\Fq$ with $\#S\geq{N}$.
\end{theorem}

\begin{theorem}
\label{thm:smallProb}
Given a positive integer $k$, there exist a threshold $N$ and $\eps>0$ such that the following holds: for every odd prime power $q>N$, if a curve $E:y^2=f(x)$ is chosen uniformly at random among all degree $4k-1$ hyperelliptic curves over $\Fq$, then with a probability at least $\eps$, we have
\[
|\#E(\Fq,S)-\#S|>0.8577\sqrt{k}\sqrt{\#S},
\]
for any set $S\subset\Fq$ with $\#S\geq{N}$.
\end{theorem}

These two theorems imply that one can expect large deviation of magnitude $\sqrt{\#S}$. In the last section, we show that for small sets $S$ of prime fields $\Fp$, the error is often much larger.

\section{Preliminaries}
\label{sec:ellipticPreliminaries}
Throughout this section, let $q$ be an odd prime power and let $n,k$ be positive integers such that $4k<n\leq{q}$. Suppose $S=\{s_1,\ldots,s_n\}\subset\Fq$, and
\[
f(x)=\sum_{j=0}^{4k-1}a_{j}x^{j}\in\Fq[x]
\]
is chosen uniformly at random. 

We denote by $\#QR$, $\#NR$ and $\#R$ the numbers of $s_i\in{S}$ such that $f(s_i)$ is an invertible quadratic residue, a quadratic non-residue and zero in $\Fq$, respectively. Then, $n=\#QR+\#NR+\#R$. It follows that, provided the curve $E:y^2=f(x)$ forms a degree $4k-1$ hyperelliptic curve over $\Fq$, the discrepancy we are interested in is
\begin{equation}
\label{eq:discrepancyQRNR}
\abs{\#E(\Fq,S)-n}=\abs{2\,\#QR+\#R-n}=\abs{\#QR-\#NR}.
\end{equation}
This suggests we look at the random variables $X_i=\legendre{f(s_i)}{q}$, where $\legendre{a}{q}$, is the Legendre symbol defined as
\[
\legendre{a}{q}=\begin{cases}
0,&\mbox{if $a$ is the zero in $\Fq$},\\
1,&\mbox{if $a$ is a non-zero square in $\Fq$},\\
-1,&\text{otherwise.}
\end{cases}
\]

We note that among all polynomials $f(x)\in\Fq[x]$ of degree at most 3, only a small fraction fail to form elliptic curves. Indeed, the exceptions, where $f(x)$ has degree strictly less than three or has multiple roots, contribute $q^3+q^2(q-1)$ of all the $q^4$ polynomials considered. When $q$ is large, such exceptions are negligible. This situation generalizes to hyperelliptic curves.

\begin{lemma}
\label{lem:constantc_k}
Let $q$ be a prime power and let $k$ be a positive integer. There is a constant $c=c_{q,k}$ such that all but at most a fraction $c/q$ of the polynomials in $\Fq[x]$ of degree at most $4k-1$ define a hyperelliptic curve of degree $4k-1$ over $\Fq$. 
\end{lemma}

Hence, the probability that, among all degree $4k-1$ hyperelliptic curves over $\Fq$, the discrepancy \eqref{eq:discrepancyQRNR} is larger than some $\delta\sqrt{n}$ is at least the probability that, among all polynomials of degree at most $4k-1$ over $\Fq$, the absolute value of the sum of the random variables $X_i$ is larger than the same $\delta\sqrt{n}$ minus $c_{q,k}/q$, i.e.
\begin{equation}
\label{eq:probabilityInequality}
\PP(\abs{\#E(\Fq,S)-n}>\delta\sqrt{n})\geq\PP\left(\abs{\sum_{i=1}^nX_i}>\delta\sqrt{n}\right)-\dfrac{c_{q,k}}{q}.
\end{equation}

In the next two subsections, we will first estimate the higher moments
\[
\EE_{j}:=\EE\left(\left(\dfrac{1}{\sqrt{n}}\sum_{i=1}^{n}X_i\right)^{j}\right),\quad\text{where }1\leq{j}\leq{4k},
\]
by finding their main order, and then give lower bounds on the interested probabilities involving the random variables $X_i$'s.

\subsection{Estimating $\EE_{2k}$ and $\EE_{4k}$}
Since $f(x)\in\Fq[x]$ is a random polynomial of degree at most $4k-1$, the random variables $X_i$ exhibit $4k$-wise independence. Indeed, by solving a system of linear equations, the number of polynomials $f(x)$ in $\Fq[x]$ of degree at most $4k-1$ satisfying
\[
f(s_{i_1})=r_1,\quad f(s_{i_2})=r_2,\quad\ldots,\quad f(s_{i_\ell})=r_\ell,
\]
is exactly $q^{4k-\ell}$, given $\ell\leq4k$, $r_1,\ldots,r_\ell\in\Fq$ and distinct $i_1,\ldots,i_\ell\in\{1,\ldots,n\}$. Thus,
\[
\begin{split}
&\quad\,\,\EE(X_{i_1}^{h_1}\cdots{X_{i_\ell}^{h_{\ell}}})\\
&=\sum_{r_1,\ldots,r_\ell\in\Fq}\PP(f(s_{i_1})=r_1,\ldots,f(s_{i_\ell})=r_\ell)\legendre{r_1}{q}^{h_1}\cdots\legendre{r_\ell}{q}^{h_\ell}\\
&=\sum_{r_1,\ldots,r_\ell\in\Fq}\dfrac{q^{4k-\ell}}{q^{4k}}\legendre{r_1}{q}^{h_1}\cdots\legendre{r_\ell}{q}^{h_\ell}\\
&=\left[\sum_{r_1\in\Fq}\dfrac{1}{q}\legendre{r_1}{q}^{h_1}\right]\cdots\left[\sum_{r_\ell\in\Fq}\dfrac{1}{q}\legendre{r_\ell}{q}^{h_\ell}\right]\\
&=\left[\sum_{r_1\in\Fq}\PP(f(s_{i_1})=r_1)\legendre{r_1}{q}^{h_1}\right]\cdots\left[\sum_{r_\ell\in\Fq}\PP(f(s_{i_\ell})=r_\ell)\legendre{r_\ell}{q}^{h_\ell}\right]\\
&=\EE(X_{i_1}^{h_1})\cdots\EE(X_{i_\ell}^{h_\ell})
\end{split}
\]

We also note that the random variables $X_i$ only take the values $0,1,-1$, and so $X_i^{2h-1}=X_i$ and $X_i^{2h}=X_i^2$, for all $h\geq1$. Therefore we have
\[
\EE(X_i^{2h-1})=\EE(X_i)=\sum_{r\in\Fq}\PP(f(s_i)=r)\legendre{r}{q}=\sum_{r\in\Fq}\dfrac{1}{q}\legendre{r}{q}=0,
\]
and
\[
\EE(X_i^{2h})=\EE(X_i^2)=\sum_{r\in\Fq}\PP(f(s_i)=r)\legendre{r}{q}^2=\sum_{r\in\Fq}\dfrac{1}{q}\legendre{r}{q}^2=1-\dfrac{1}{q}.
\]

To summarize the above two observations, we have the following lemma:
\begin{lemma}
\label{lem:expectationOddEven}
Let $\ell\leq4k$, let $h_1,\ldots,h_\ell$ be positive integers, and let $i_1\,\ldots,i_\ell$ be distinct numbers from $\{1,\ldots,n\}$. Then,
\[
\EE(X_{i_1}^{h_1}\cdots{X_{i_\ell}^{h_{\ell}}})=
\begin{cases}
\left(1-\dfrac{1}{q}\right)^\ell,&\mbox{if $h_1,\ldots,h_\ell$ are all even numbers},\\
0,&\text{otherwise.}
\end{cases}
\]
\end{lemma}

Before we estimate the general $\EE_{j}$, let us compute $\EE_6$ (when $k\geq2$) as a toy version.
\[
\begin{split}
\EE_6
&=\EE\left(\dfrac{1}{\sqrt{n}}\sum_{i=1}^{n}X_i\right)^{6}\\
&=\dfrac1{n^3}\left(\sum_{i=1}^n\EE(X_i^6)+\dfrac{6!}{4!2!}\sum_{i\neq{j}}\EE(X_i^4X_j^2)+\dfrac{6!}{2!2!2!}\sum_{i<j<k}\EE(X_i^2X_j^2X_k^2)\right)\\
&=\dfrac1{n^3}\left(n\left(1-\dfrac1q\right)+15n(n-1)\left(1-\dfrac1q\right)^2+90\binom{n}{3}\left(1-\dfrac1q\right)^3\right)\\
&=15\left(1-\dfrac1q\right)^3-\dfrac{15}{n}\left(1-\dfrac1q\right)^2\left(2-\dfrac3q\right)+\dfrac{1}{n^2}\left(1-\dfrac1q\right)\left(16-\dfrac{45}{q}+\dfrac{30}{q^2}\right)
\end{split}
\]
We derive in the lemma below how the number $15$ in the leading term can be expressed in terms of $j=6$.

\begin{lemma}
\label{lem:higherMoments}
For $1\leq{j}\leq{4k}$, we have
\[
\EE_{j}=
\begin{cases}
\dfrac{j!}{2^{j/2}(j/2)!}+O_j\left(\dfrac1n\right),\quad\text{as}\quad{n}\to\infty,&\mbox{if $j$ is an even number,}\\
0,&\text{otherwise.}
\end{cases}
\]
\end{lemma}

\begin{proof}
If $j$ is an odd number, then every term in the multinomial expansion has at least one odd index, and hence vanishes by Lemma \ref{lem:expectationOddEven}.

Suppose now that $j$ is an even integer. Using the multinomial theorem and Lemma \ref{lem:expectationOddEven}, we have
\[
\begin{split}
\EE_j
&=\dfrac{1}{n^{j/2}}\EE\left(\left(\sum_{i=1}^nX_i\right)^j\right)\\
&=\dfrac{1}{n^{j/2}}\EE\left(\sum_{h_1+\cdots+h_n=j}\dfrac{j!}{h_1!\cdots{h_n!}}\prod_{t=1}^nX_t^{h_t}\right)\\
&=\dfrac{1}{n^{j/2}}\sum_{h_1+\cdots+h_n=j}\dfrac{j!}{h_1!\cdots{h_n!}}\EE\left(\prod_{t=1}^nX_t^{h_t}\right)\\
&=\dfrac{1}{n^{j/2}}\sum_{\substack{h_1+\cdots+h_n=j\\h_i\text{ even}}}\dfrac{j!}{h_1!\cdots{h_n!}}\left(1-\dfrac{1}{q}\right)^{\#\{i:h_i>0\}}\\
&=\dfrac{1}{n^{j/2}}\sum_{m=1}^{j/2}\left(1-\dfrac{1}{q}\right)^mH(j,m),
\end{split}
\]
where
\[
H(j,m)=\sum_{\substack{h_1+\cdots+h_n=j\\h_i\text{ even}\\\#\{i:h_i>0\}=m}}\dfrac{j!}{h_1!\cdots{h_n!}}=\binom{n}{m}\sum_{\substack{h'_1+\cdots+h'_m=j\\h'_i>0\text{ even}}}\dfrac{j!}{h'_1!\cdots{h'_m!}}
\]
is a polynomial (with integer coefficients) in $n$ of degree $m$. Therefore, the leading term of $\EE_j$ comes from the summand where $m=j/2$. In this case, $h'_i=2$ for every $1\leq{i}\leq{j/2}$ and so
\[
H(j,j/2)=\binom{n}{j/2}\dfrac{j!}{2^{j/2}}
\]
has leading term
\[
\dfrac{j!}{(j/2)!2^{j/2}}n^{j/2}.
\]
It follows that
\[
\begin{split}
\EE_j
&=\dfrac{1}{n^{j/2}}\left(\left(1-\dfrac{1}{q}\right)^{j/2}\dfrac{j!}{(j/2)!2^{j/2}}n^{j/2}+\cdots\right)\\
&=\left(1-\dfrac{1}{q}\right)^{j/2}\dfrac{j!}{(j/2)!2^{j/2}}+O_j\left(\dfrac{1}{n}\right)\\
&=\dfrac{j!}{(j/2)!2^{j/2}}+O_j\left(\dfrac{1}{n}\right),
\end{split}
\]
as $n\to\infty$.
\end{proof}

In particular, for each fixed $k$,
\[
\EE_{2k}=\dfrac{(2k)!}{2^kk!}+O_k\left(\dfrac{1}{n}\right)
\]
is bounded uniformly in $n\geq1$. As a consequence, one can have the following estimates, which will be used later in our proof, using Stirling's approximation. For all fixed $k\geq1$, we have
\begin{equation}
\label{eq:2kRootE_2k}
\sqrt[2k]{\EE_{2k}}\geq\sqrt{\dfrac{2k}{e}}+O_k\left(\dfrac1n\right),
\end{equation}
and
\begin{equation}
\label{eq:E2k^2OverE4k}
\dfrac{\EE_{2k}^2}{\EE_{4k}}\geq\left(\dfrac{\sqrt{2\pi}}{e}\right)^32^{1/2-2k}+O_k\left(\dfrac1n\right),
\end{equation}
as $n\to\infty$.

\subsection{Lower bounding the probabilities}

\begin{proposition}
\label{prop:elliptic}
Under the setting stated in the beginning of this section, we have
\begin{equation}
\label{eq:propLargeProb}
\PP\left(\left|\dfrac{1}{\sqrt{n}}\sum_{i=1}^{n}X_i\right|>\delta\right)\geq\dfrac{(\EE_{2k}-\delta^{2k})^2}{\EE_{4k}-2\delta^{2k}\EE_{2k}+\delta^{4k}},
\end{equation}
for any  $0<\delta<1/2$, and
\begin{equation}
\label{eq:propSmallProb}
\PP\left(\left|\dfrac{1}{\sqrt{n}}\sum_{i=1}^{n}X_i\right|\geq\sqrt[2k]{\EE_{2k}}-\eps^{\frac12-o(1)}\right)\geq\eps>0,
\end{equation}
as $\eps\to0$.
\end{proposition}

\begin{proof}
Let $c\geq1$ be a parameter to be determined. Using second moment Markov's inequality, one can show that for $0<\lambda<{c}^{2k}$,
\begin{equation}
\label{eq:markov}
\begin{split}
\mathbb{P}\left(\left|\dfrac{1}{\sqrt{n}}\sum_{i=1}^{n}X_i\right|>\sqrt[2k]{c^k-\sqrt{\lambda}}\right)
&=\mathbb{P}\left(\left(\dfrac{1}{\sqrt{n}}\sum_{i=1}^{n}X_i\right)^{2k}-c^{k}>-\sqrt{\lambda}\right)\\
&\geq\mathbb{P}\left(\left|\left(\dfrac{1}{\sqrt{n}}\sum_{i=1}^{n}X_i\right)^{2k}-c^k\right|<\sqrt\lambda\right)\\
&\geq1-\dfrac{1}{\lambda}\EE\left(\left(\left(\dfrac{1}{\sqrt{n}}\sum_{i=1}^{n}X_i\right)^{2k}-c^k\right)^2\right)\\
&=1-\dfrac{c^{2k}-2c^{k}\EE_{2k}+\EE_{4k}}{\lambda}.
\end{split}
\end{equation}

To prove \eqref{eq:propLargeProb}, we take $\lambda=(c^k-\delta^{2k})^{2}$, where $\delta>0$ is small. Maximizing the right hand side of \eqref{eq:markov} over $c$, we see that the maximum is
\[
1-\dfrac{c^{2k}-2c^{k}\EE_{2k}+\EE_{4k}}{(c^k-\delta^{2k})^2}=\dfrac{(\EE_{2k}-\delta^{2k})^2}{\EE_{4k}-2\delta^{2k}\EE_{2k}+\delta^{4k}},
\]
when
\[
c^k=\dfrac{\EE_{4k}-\delta^{2k}\EE_{2k}}{\EE_{2k}-\delta^{2k}}.
\]

Now we prove \eqref{eq:propSmallProb}. To make
\[
\mathbb{P}\left(\left|\dfrac{1}{\sqrt{n}}\sum_{i=1}^{n}X_i\right|>\sqrt[2k]{c^{k}-\sqrt{\lambda}}\right)\geq\eps,
\]
we take
\[
\lambda=\dfrac{c^{2k}-2c^{k}\EE_{2k}+\EE_{4k}}{1-\eps}.
\]
Since we require $c^{2k}>\lambda$, it follows that
\[
c^{2k}-2c^k\EE_{2k}+\EE_{4k}<c^{2k}-c^{2k}\eps,
\]
and therefore
\[
\eta:=\eps{c^k}<2\EE_{2k}-\dfrac{\EE_{4k}}{c^k}<2\EE_{2k}.
\]
To compute the leading terms of $\sqrt[2k]{c^k-\sqrt{\lambda}}$ as $\eps\to0$, we first use the binomial series to expand the numerator of $\sqrt{\lambda}$ as
\begin{equation}
\label{eq:numeratorOfLambda}
\begin{split}
c^{k}\sqrt{1-\left(\dfrac{2\EE_{2k}}{c^k}-\dfrac{\EE_{4k}}{c^{2k}}\right)}
=c^{k}\left(1-\EE_{2k}\dfrac{1}{c^k}+\dfrac{\EE_{4k}-\EE_{2k}^2}{2}\dfrac{1}{c^{2k}}+O\left(\dfrac{1}{c^{3k}}\right)\right),
\end{split}
\end{equation}
as $c\to\infty$. Indeed, the bracket inside the square root in \eqref{eq:numeratorOfLambda} is small in view of Lemma \ref{lem:higherMoments}. To get $\sqrt\lambda$, we multiply \eqref{eq:numeratorOfLambda} to
\[
\dfrac{1}{\sqrt{1-\eps}}=1+\dfrac12\eps+\dfrac38\eps^2+O(\eps^3).
\]
Substituting $c^{k} = \eta/\eps$, we have
\[
\begin{split}
&\quad\,\,c^{k}-\sqrt{\lambda}\\
&=\dfrac{\eta}{\eps}\left[1-\left(1+\dfrac12\eps+\dfrac38\eps^2+O(\eps^3)\right)\left(1-\dfrac{\EE_{2k}}{\eta}\eps+\dfrac{\EE_{4k}-\EE_{2k}^2}{2\eta^2}\eps^2+O\left(\dfrac{\eps^3}{\eta^3}\right)\right)\right]\\
&=\EE_{2k}-\dfrac{1}{2}\eta+\left(\dfrac{\EE_{2k}^2-\EE_{4k}}{2}+\dfrac{\EE_{2k}}{2}\eta-\dfrac{3}{8}\eta^{2}\right)\dfrac{\eps}{\eta}+O\left(\dfrac{\eps^2}{\eta^2}\right).
\end{split}
\]
We may now take $\eta$ satisfying $\sqrt\eps\ll\eta\ll1$ so that the terms in the last line are indeed arranged in decreasing order of magnitude. Therefore,
\[
\sqrt[2k]{c^k-\sqrt\lambda}=\sqrt[2k]{\EE_{2k}-\eps^{\frac12-o(1)}}=\sqrt[2k]{\EE_{2k}}-\eps^{\frac12-o(1)},
\]
as $\eps\to0$, establishing \eqref{eq:propSmallProb}.
\end{proof}

\section{Proofs of the Theorems}

\begin{proof}[Proof of Theorem \ref{thm:largeProb}]
Write $n=\#S$, as in Section \ref{sec:ellipticPreliminaries}. Given $\eps>0$, we choose $N$ large enough so that $c_{q,k}/N<\eps/3$, where the constant $c_{q,k}$ is from Lemma \ref{lem:constantc_k}, and the error appearing in \eqref{eq:E2k^2OverE4k} has an absolute value less than $\eps/3$.

Since $\EE_{4k}>\EE_{2k}\geq1/2$, there exists a small $\delta>0$ such that
\[
\left|\dfrac{(1-\frac{\delta^{2k}}{\EE_{2k}})^2}{1-2\delta^{2k}\frac{\EE_{2k}}{\EE_{4k}}+\delta^{4k}\frac{1}{\EE_{4k}}}-1\right|<\dfrac{\eps}{3}\dfrac{\EE_{4k}}{\EE_{2k}^2}.
\]
Together with \eqref{eq:probabilityInequality}, \eqref{eq:propLargeProb} and \eqref{eq:E2k^2OverE4k}, we have
\[
\begin{split}
\PP(\abs{\#E(\Fq,S)-n}>\delta\sqrt{n})
&\geq\PP\left(\abs{\sum_{i=1}^nX_i}>\delta\sqrt{n}\right)-\dfrac{c_{q,k}}{q}\\
&\geq\dfrac{\EE_{2k}^2}{\EE_{4k}}\dfrac{(1-\frac{\delta^{2k}}{\EE_{2k}})^2}{1-2\delta^{2k}\frac{\EE_{2k}}{\EE_{4k}}+\delta^{4k}\frac{1}{\EE_{4k}}}-\dfrac{\eps}{3}\\
&\geq\dfrac{\EE_{2k}^2}{\EE_{4k}}-\dfrac{\eps}{3}-\dfrac{\eps}{3}\\
&\geq\left(\dfrac{\sqrt{2\pi}}{e}\right)^32^{1/2-2k}-\eps,
\end{split}
\]
as desired.
\end{proof}

\begin{proof}[Proof of Theorem \ref{thm:smallProb}]
Similarly we write $n=\#S$. Using the estimate \eqref{eq:2kRootE_2k}, we choose $N$ so large and $\eps$ so small that the interested lower bound in \eqref{eq:propSmallProb} is large:
\[
\sqrt[2k]{\EE_{2k}}-\eps^{\frac12-o(1)}>0.8577\sqrt{k}.
\]
Here $0.8577$ is a number strictly smaller than $\sqrt{2/e}$. Now, increase $N$ if necessary, we also have $c_{q,k}/N<\eps/2$. Then, by \eqref{eq:probabilityInequality} and \eqref{eq:propSmallProb}, we have
\[
\begin{split}
&\quad\,\,\PP(\abs{\#E(\Fq,S)-n}>0.8577\sqrt{k}\sqrt{n})\\
&\geq\PP\left(\abs{\sum_{i=1}^nX_i}>0.8577\sqrt{k}\sqrt{n}\right)-\dfrac{c_{q,k}}{q}\\
&\geq\PP\left(\abs{\sum_{i=1}^nX_i}>\left(\sqrt[2k]{\EE_{2k}}-\eps^{\frac12-o(1)}\right)\sqrt{n}\right)-\dfrac{\eps}{2}\\
&\geq\dfrac{\eps}{2}.
\end{split}
\]
\end{proof}

\section{Sets with exceptionally large discrepancy}

So far we have considered sets of arbitrarily large size. We will show, as one may expect, that if $n$ is a constant, then for each prime $p$ large enough, there is a probability $\alpha>0$ that the error is much larger than $\sqrt{n}$, for $\beta\binom{p}{n}$ of the subsets $S\subset\Fp$ of size $n$. In particular, for each $n$, there is a probability $2^{-n-1}$ that a randomly chosen subset $S\subset\Fp$ of size $n$ has the following property --- a randomly chosen monic separable cubic $f$ over $\Fp$ has a probability $2^{-n-1}$ so that $f(S)$ consists only of non-zero quadratic residues or quadratic non-residues.\\

Let $\mathcal F$ be the set of monic, separable cubics over $\Fp$. Note that $\#\mathcal F = p^3-p^2$. Let $m,n$ be constants independent of $p$ such that $n-2m>\sqrt{n}$. We construct a bipartite graph $G$ with $\binom{p}{n}$ `S-vertices' in one partition, each associated with a set $S \subset \Fp$ of size $n$, and $p^3 - p^2$ `f-vertices' in the other, each associated with an $f \in \mathcal F$. We draw an edge between the vertex corresponding to $f$ and the vertex corresponding to $S$ when $\left| \sum_{s_i \in S} \legendre{f(s_i)}{p} \right| \geq n-2m$. Fix $f \in \mathcal F$, and let $\mathcal Q\subset\Fp$ be the set of points mapped by $f$ to a non-zero quadratic residue, and $\mathcal N\subset\Fp$ be those points mapped to a non-residue. Let $p/2+A_f$ be the size of the larger of these two sets. Then the degree of the vertex associated to $f$ in $G$ is at least
\begin{equation}
\label{eq:minDegree}
 \binom{p/2 - A_f}{m} \binom{ p/2 + A_f }{n-m}.
\end{equation}
By Hasse's theorem we have $A_f \leq \sqrt{p}$, and so \eqref{eq:minDegree} is bounded below by
\[
\binom{p/2 - \sqrt{p} }{m} \binom{p/2 -\sqrt{p} }{n-m}=\binom{p}{n}\left[\binom{n}{m}2^{-n}+o(1)\right],
\]
as $p\to\infty$.
Thus the number of edges in our graph, $E$, is at least
\[
\binom{p}{n}\left[\binom{n}{m}2^{-n}+o(1)\right](p^3 - p^2).
\]
Now if only $\beta\binom{p}{n}$ of the $S$-vertices achieve degree $\geq \alpha(p^3 - p^2)$, then we have
\[
E \leq \beta \binom{p}{n}(p^3-p^2) + \binom{p}{n}(1-\beta)\alpha (p^3 - p^2),
\]
and so
\[
\beta\geq\dfrac{1}{1-\alpha}\left[\binom{n}{m}2^{-n}-\alpha+o(1)\right]>0,
\]
as $p\to\infty$, provided that $\alpha>0$ is small enough.

\bibliographystyle{acm}
\bibliography{ref_elliptic}

\begin{thebibliography}{1}

\bibitem{handbook06}
{\sc Cohen, H., Frey, G., Avanzi, R., Doche, C., Lange, T., Nguyen, K., and
  Vercauteren, F.}, Eds.
\newblock {\em Handbook of elliptic and hyperelliptic curve cryptography}.
\newblock Discrete Mathematics and its Applications (Boca Raton). Chapman \&
  Hall/CRC, Boca Raton, FL, 2006.

\bibitem{Has36}
{\sc Hasse, H.}
\newblock Zur {T}heorie der abstrakten elliptischen {F}unktionenk\"{o}rper
  {III}. {D}ie {S}truktur des {M}eromorphismenrings. {D}ie {R}iemannsche
  {V}ermutung.
\newblock {\em J. Reine Angew. Math. 175\/} (1936), 193--208.

\bibitem{HSS01}
{\sc Hess, F., Seroussi, G., and Smart, N.~P.}
\newblock Two topics in hyperelliptic cryptography.
\newblock In {\em Selected areas in cryptography}, vol.~2259 of {\em Lecture
  Notes in Comput. Sci.} Springer, Berlin, 2001, pp.~181--189.

\bibitem{PelRam17}
{\sc Pelekis, C., and Ramon, J.}
\newblock Hoeffding's {I}nequality for {S}ums of {D}ependent {R}andom
  {V}ariables.
\newblock {\em Mediterr. J. Math. 14}, 6 (2017), 14:243.

\bibitem{Pol75}
{\sc Pollard, J.~M.}
\newblock A {M}onte {C}arlo method for factorization.
\newblock {\em Nordisk Tidskr. Informationsbehandling (BIT) 15}, 3 (1975),
  331--334.

\bibitem{RSA78}
{\sc Rivest, R.~L., Shamir, A., and Adleman, L.}
\newblock A method for obtaining digital signatures and public-key
  cryptosystems.
\newblock {\em Comm. ACM 21}, 2 (1978), 120--126.

\bibitem{Sat09}
{\sc Satoh, T.}
\newblock Generating genus two hyperelliptic curves over large characteristic
  finite fields.
\newblock In {\em Advances in cryptology---{EUROCRYPT} 2009}, vol.~5479 of {\em
  Lecture Notes in Comput. Sci.} Springer, Berlin, 2009, pp.~536--553.

\bibitem{SSS95}
{\sc Schmidt, J.~P., Siegel, A., and Srinivasan, A.}
\newblock Chernoff-{H}oeffding bounds for applications with limited
  independence.
\newblock {\em SIAM J. Discrete Math. 8}, 2 (1995), 223--250.

\bibitem{Sem98}
{\sc Semaev, I.~A.}
\newblock Evaluation of discrete logarithms in a group of {$p$}-torsion points
  of an elliptic curve in characteristic {$p$}.
\newblock {\em Math. Comp. 67}, 221 (1998), 353--356.

\end{thebibliography}

\end{document}